\documentclass[a4paper, 11pt]{article}
\usepackage[T1]{fontenc}
\usepackage{amsmath, amsfonts}
\usepackage{booktabs, multirow}
\usepackage{subcaption}
\usepackage[margin=1in]{geometry}

\usepackage{algorithm, algpseudocode}
\algnewcommand\algorithmicinput{\textbf{Input:}}
\algnewcommand\algorithmicoutput{\textbf{Output:}}
\algnewcommand\Input{\item[\algorithmicinput]}
\algnewcommand\Output{\item[\algorithmicoutput]}

\usepackage{setspace}
\usepackage{mathtools} 

\usepackage{amsthm}
\newtheorem{theorem}{Theorem}
\newtheorem{proposition}{Proposition}
\newtheorem{lemma}{Lemma}
\newtheorem{corollary}{Corollary}
\newtheorem{remark}{Remark}

\usepackage{tikz}
\usetikzlibrary{arrows, shapes}
\usetikzlibrary{decorations.pathmorphing}
\usetikzlibrary{calc, arrows.meta, positioning}
\usetikzlibrary{shapes.geometric}

\usepackage{apacite}
\usepackage{natbib}
\usepackage{hyperref}
\usepackage{cleveref}
\usepackage{authblk}

\title{Centrality of shortest paths: algorithms and complexity results}
\author{Johnson Phosavanh} 
\author{Dmytro Matsypura}
\affil{The University of Sydney}

\begin{document}

\maketitle

\begin{abstract}
The centrality of a node is often used to measure its importance to the structure of a network. Some centrality measures can be extended to measure the importance of a path. In this paper, we consider the problem of finding the most central shortest path. We show that the computational complexity of this problem depends on the measure of centrality used and, in the case of degree centrality, whether the network is weighted or not. We develop a polynomial algorithm for the most degree-central shortest path problem with the worst-case running time of $O(|E||V|^2\Delta(G))$, where $|V|$ is the number of vertices in the network, $|E|$ is the number of edges in the network, and $\Delta(G)$ is the maximum degree of the graph. In addition, we show that the same problem is NP-hard on a weighted graph. Furthermore, we show that the problem of finding the most betweenness-central shortest path is solvable in polynomial time, while finding the most closeness-central shortest path is NP-hard, regardless of whether the graph is weighted or not. We also develop an algorithm for finding the most betweenness-central shortest path with a running time of $O(|E|^2|V|^2)$ on both weighted and unweighted graphs. To conclude our paper, we conduct a numerical study of our algorithms on synthetic and real-world networks and compare our results to the existing literature.
\end{abstract}

\section{Introduction}

\subsection{Motivation}

The concept of centrality plays a vital role in understanding network structure. It allows us to measure the relative importance of nodes in a network. For instance, in social networks, the more central a node is, the better information can spread from and through it. Various metrics have been proposed to measure the centrality of nodes, such as betweenness centrality, closeness centrality, degree centrality, and eigenvector centrality, to name a few. In their seminal work, \citet{Everett1999} propose to extend some of these metrics and consider the centrality of a group of nodes. The authors extend the definitions of degree, closeness, and betweenness centrality to groups of nodes to quantitatively study their influence. With these measures, a natural question arises: what is the most central group of nodes with a pre-defined structure? The first attempts to look at this class of problems date back to the early 1980's \citep{MitchelCockayne1981, Slater1982}. In these papers, the authors restrict the group of nodes to be a simple path and consider the problem of finding paths of minimum eccentricity and minimum distance in a tree. In both cases, the authors develop linear algorithms to solve the problem. \citet{Gomez2023} extend the work of \citet{MitchelCockayne1981} and \citet{Slater1982}. The authors restrict their analysis to path eccentricity, and provide upper bounds for the minimum eccentricity for arbitrary and $k$-connected graphs.

\citet{dawande2012} introduce the portal problem, where the objective is to find a set of at most $k$ nodes with the highest normalised betweenness centrality score. The authors show that the portal problem is strongly NP-hard. \citet{Vogiatzis2014} investigate the problem of finding the most and the least central cliques in a graph. The authors formulate linear 0-1 programs and examine their performance on real-world and synthetic instances. \citet{Vogiatzis2019} study a similar problem of finding the most degree-central induced star. They propose an integer program and heuristic greedy algorithms and apply them to find the most essential proteins in protein-protein interaction networks. \citet{Camur2021} extend this work by developing faster linear 0-1 programs and presenting classes of networks in which the problem is solvable in polynomial time.

\citet{Matsypura2023} investigate the problem of finding the most degree-central walks in a graph. The authors consider problems of finding the most degree-central walk, trail, path, induced path, and shortest path. They show that the problems are nested in the sense that a shortest path is an instance of an induced path, an induced path is an instance of a simple path, which is an instance of a trail, which, in turn, is an instance of a walk. Further, all the problems considered in that study were shown to be NP-hard, except for the problem of finding the most degree-central shortest path, which can be solved in polynomial time using an algorithm based on the breadth-first search technique. We will refer to this algorithm as the MVP algorithm. For all the remaining problems, the authors developed mixed-integer programs and proposed heuristic solutions as a warm start to improve performance. For a more thorough review of group centrality optimisation problems, we refer the reader to a recent survey by \citet{Camus2024}, who conclude that ``there is a research gap for these novel group [walk and path group] centrality metrics with respect to closeness and betweenness centrality that we believe needs to be addressed.''

In this paper, we address the research gap identified by \citet{Camus2024} and focus on centrality of shortest paths. We first develop an efficient polynomial algorithm to find the most degree-central shortest path. Our computational experiments demonstrate the advantage of the proposed algorithm relative to the MVP algorithm. We also show that the same problem on a weighted graph is NP-hard even in the presence of just two distinct weights. We then study the problem of finding the most betweenness-central shortest path and show that it is solvable in polynomial time for both weighted and unweighted cases. Finally, we prove that the problem of finding the most closeness-central shortest path is NP-hard regardless of whether the weights are present. 

The class of problems we are considering extends the classical facility location problem, where a central planner seeks to determine which set of factories or warehouses to open to minimise transportation costs or maximise revenue. In our case, we apply the restriction that the selected facilities must also form a shortest path between the start and end nodes. This class of problems is similar to the hub line location problem proposed by \citet{deSa2015}, where the goal is to connect cities along a single path while maximising the network's reach. The hub line location problem can be used to determine routes for efficient and direct high-capacity transport systems such as rail, which may be supplemented by trucks and drones for last-mile delivery, or by buses and bike share systems in the case of public transportation \citep{To2015}. During natural disasters, authorities may use this model to prioritise locations for setting up evacuation centres by maximising the number of affected communities that can be reached most efficiently. The hub line location problem can also be used to plan fibre optic cable connections \citep{GutierrezJarpa2010}. In this problem, a planner aims to establish a fibre optic network interconnecting WiFi antennas or individual local users. In all the infrastructure applications mentioned above, the cost of establishing a path can be quite high, naturally leading to the question: what is the most cost-effective way to create a path that maximises the network's reach? Similarly, in social networks, by considering each user as a node, we can apply our model to identify the most influential users that can be reached most efficiently \citep{Peng2017}. This information can be used effectively by an adversarial player when distributing (mis)information and software with malicious intent.

\subsection{Problem definition}

Let $G = (V, E)$ be an unweighted (possibly directed) graph with a set of vertices $V$ and a set of edges $E$. Without loss of generality, we assume that $G$ is connected if it is undirected and weakly connected if it is directed. Following the definitions provided by \citet{Matsypura2023}, we use $P$ to denote a path, i.e., $P$ is a finite sequence of distinct adjacent vertices in $G$. We write $P = \langle v_1, v_2, \dots, v_n \rangle$ for some $v_1, \ldots, v_n \in V$ when we need to emphasise the vertices on the path. We also use $\mathcal{SP}(G)$ to denote the set of all shortest paths between all pairs of nodes in $G$. 

Given a graph $G$ and a measure of centrality $C(P)$, we aim to solve the following problem:
\begin{equation}
\label{eq:problem}
\max \{C(P) \colon P \in \mathcal{SP}(G)\}.
\end{equation}
In other words, we seek to find a path $P$ with the largest centrality under the condition that $P$ is a shortest path between a pair of nodes in $G$. Problem~\eqref{eq:problem} is closely related to the Central Path problem, introduced by \citet{MitchelCockayne1981}, with two notable differences. First, in the Central Path problem, the aim is to find a simple path, whereas our path of interest must be shortest. Second, in the Central Path problem, the measure of centrality is restricted to eccentricity, whereas we do not impose such restriction. Throughout the paper, we use three different measures of centrality: degree, betweenness, and closeness. We define each of them as it becomes necessary.

\subsection{Structure of the paper}

The remaining paper is structured as follows: in \Cref{sec:algorithm}, we study Problem~\eqref{eq:problem} with degree centrality as a measure of centrality $C(P)$ and present a polynomial-time algorithm that can be used to solve it.  
In \Cref{sec:weighted}, we extend our analysis to the case of graphs with weighted edges and show that the problem becomes NP-hard even with just two distinct weights. In \Cref{sec:btw}, we consider Problem~\eqref{eq:problem} with betweenness centrality as a measure of centrality $C(P)$. We show that this problem is solvable in polynomial time for both weighted and unweighted graphs and present a polynomial-time algorithm that can be used to solve it. In \Cref{sec:cls}, we turn our attention to closeness centrality. We prove that Problem~\eqref{eq:problem} with closeness centrality as a measure of centrality $C(P)$ is NP-hard regardless of whether the graph is weighted or not. This is followed by \Cref{sec:computational-experiments}, where we conduct an extensive numerical study of the proposed algorithms on synthetic and real-world graph instances and provide insights into the results. Finally, we provide concluding remarks in \Cref{sec:conclusion}.

\section{Degree centrality}\label{sec:algorithm}
We first consider the problem of finding the most degree-central shortest path, which we define as follows. Let $\mathcal{N}(P)$ denote the neighbourhood of path $P$, i.e., the set of vertices adjacent to $P$, excluding the vertices of $P$ itself:
\[
\mathcal{N}(P) := \{v \colon (u, v) \in E, u \in P\} \setminus P.
\]
For path $P$, we define \emph{path degree centrality} $C_d(P)$ to be the number of vertices adjacent to it, i.e., $C_d(P) := |\mathcal{N}(P)|$. For example, path degree centrality can be used to represent the number of users in a computer network that are immediately susceptible to an attack by an adversarial player who has control over a series of computers. Our definition is consistent with the definition of \emph{group degree centrality} in \cite{Everett1999}. Note that when the path consists of a single node, path degree centrality reduces to the traditional \emph{node degree centrality}. It then follows that the problem of finding the most degree-central shortest path becomes
\begin{equation}\label{eq:degree-problem}
    \max \{C_d(P) \colon P \in \mathcal{SP}(G)\}.
\end{equation} 

Before formally introducing the algorithm, we prove the following lemma that underpins many efficient algorithms for solving problems involving shortest paths.

\begin{lemma}\label{lem:shortest-path}
    If $\langle s, p_1, \ldots, p_{k}, t \rangle$ is a shortest path from $s$ to $t$, then $\langle s, p_1, \ldots, p_{k} \rangle$ is a shortest path from $s$ to $p_{k}$.
\end{lemma}

\begin{proof}
    Suppose there is a shorter path from $s$ to $p_{k}$, and let this be $\langle s, \tilde{p}_1, \ldots, \tilde{p}_{j},  p_k \rangle$. This implies that $\langle s, \tilde{p}_1, \ldots, \tilde{p}_{j}, p_{k}, t\rangle$ would be a shorter path from $s$ to $t$. However, this cannot be true as $\langle s, p_1, \ldots, p_{k}, t \rangle$ is a shortest path from $s$ to $t$.
\end{proof}

We now extend \Cref{lem:shortest-path} to path degree centrality.

\begin{lemma}\label{lem:most-central-path}
    If $\langle s, p_1, \ldots, p_{k - 1}, p_{k}, t \rangle$ is a most degree-central shortest path from $s$ to $t$, then $\langle s, p_1, \ldots, p_{k - 1} \rangle$ is a most degree-central shortest path from $s$ to $p_{k - 1}$.
\end{lemma}

\begin{proof}
    First, observe that vertex $t$ cannot share any neighbours with $\langle s, p_1, \ldots, p_{k - 2}\rangle$. To see this, assume that vertex $p_i$ on the path, where $i \leq k - 2$, shares a neighbour $v$ with vertex $t$. Then $\langle s, p_1, \ldots, p_{k - 1}, p_{k}, t \rangle$ cannot possibly be the most degree-central shortest path between $s$ and $t$ as $\langle s, p_1, \ldots, p_{i}, v, t \rangle$ would be a shorter path.

    Now consider the graph in \Cref{fig:most-degree-central-structure} where vertex $t$ does not share any neighbours with any node on path $\langle s, p_1, \ldots, p_{k - 2}\rangle$, but with $p_{k - 1}$ and $p_k$.
    
    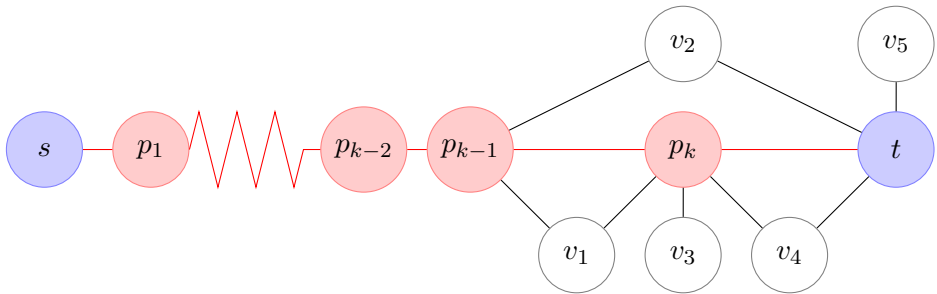
\begin{figure}[H]
        \centering
        \begin{tikzpicture}[scale=1.4]
            \clip (-0.5,-1.5) rectangle (8.5,1.5);
    
            \begin{scope}[circle,minimum size=10mm]
                \node at (0, 0) [draw=blue!50,fill=blue!20] (s) {$s$};
                \node at (1, 0) [draw=red!50,fill=red!20] (1) {$p_1$};
                \node at (3, 0) [draw=red!50,fill=red!20] (k-2) {$p_{k - 2}$};
                \node at (4, 0) [draw=red!50,fill=red!20] (k-1) {$p_{k - 1}$};
                \node at (6, 0) [draw=red!50,fill=red!20] (k) {$p_{k}$};
                \node at (5, -1) [draw=black!50] (v1) {$v_1$};
                \node at (6, 1) [draw=black!50] (v2) {$v_2$};
                \node at (6, -1) [draw=black!50] (v3) {$v_3$};
                \node at (7, -1) [draw=black!50] (v4) {$v_4$};
                \node at (8, 1) [draw=black!50] (v5) {$v_5$};
                \node at (8, 0) [draw=blue!50,fill=blue!20] (t) {$t$};
            \end{scope}
            
            \path [every node/.style={font=\sffamily\small, sloped}]
            (s) edge [red] node [above] {} (1)
            (k-1) edge [red] node [above] {} (k)
            (k-2) edge [red] node [above] {} (k-1)
            (k-1) edge node [above] {} (v1)
            (k-1) edge node [above] {} (v2)
            (v2) edge node [above] {} (t)
            (v1) edge node [above] {} (k)
            (k) edge [red] node [above] {} (t)
            (k) edge node [above] {} (v3)
            (k) edge node [above] {} (v4)
            (v4) edge node [above] {} (t)
            (v5) edge node [above] {} (t);
            
            \draw [decorate, decoration={zigzag, segment length=5mm, amplitude=5mm}, red] (1) -- (k-2);
            
        \end{tikzpicture}
        \caption{Illustration of \Cref{lem:most-central-path}. The path $\langle s, p_1, \ldots, p_{k - 1} \rangle$ must be a most degree-central shortest path between $s$ and $p_{k - 1}$ if the path $\langle s, p_1, \ldots, p_{k}, t\rangle$ is a most degree-central shortest path between $s$ and $t$.}
        \label{fig:most-degree-central-structure}
    \end{figure}
    
    From the diagram, it is clear that if $\langle s, p_1, \ldots, p_{k - 1}, p_{k}, t \rangle$ is a most degree-central shortest path from $s$ to $t$, then $\langle s, p_1, \ldots, p_{k - 1} \rangle$ is a most degree-central shortest path from $s$ to $p_{k - 1}$. To see this, suppose that there is an alternate shortest path $\langle s, \tilde{p}_1, \ldots, \tilde{p}_{k - 2}, p_{k - 1}\rangle$ to vertex $p_{k - 1}$ that shares no neighbours with vertex $t$ and has a higher degree centrality. In this case, path $\langle s, \tilde{p}_1, \ldots, \tilde{p}_{k - 2}, p_{k - 1}, p_k, t\rangle$, would have a degree centrality greater than that of path $\langle s, p_1, \ldots, p_{k - 1}, p_{k}, t \rangle$ as vertices $p_k$ and $t$ contribute the same vertices to the neighbourhoods of both paths.
\end{proof}

We now introduce the following notation to simplify the exposition. First, consider a source node $s$. Let $d_u$ and $P_u$ be the shortest distance and most degree-central shortest path between nodes $s$ and $u$, respectively, and let $\mathcal{P}_u$ be the set of vertices that can directly precede vertex $u$ on any shortest path from $s$ to $u$. Additionally, we introduce the operation $\langle P, u, v\rangle$, which we use to denote extending the path $P$ to vertices $u$, then $v$.

To solve Problem~\eqref{eq:problem}, we extend the breadth-first search (BFS) algorithm. The BFS algorithm can be used to traverse a graph from a source node $s$ to any other node in the network by iteratively extending paths already constructed in previous iterations. Specifically, we begin with the source node $s$ and, at each iteration of the algorithm, consider a queue of vertices $Q$ such that $q \in Q$, if a path from $s$ to $q$ has been established in a previous iteration and $q$ hasn't been processed yet. Then, we pop the first element in $Q$, denoting it $u$, and consider the set  $\{(u, v) \colon (u, v)\in E\}$, i.e., edges that can be connected to $u$. We then check if $v$ has been processed in a previous iteration, and if it has not, we extend the path to $v$ and add $v$ to the end of $Q$ (this is the traditional BFS algorithm). Note that this procedure also identifies the shortest path from a source node $s$ to all other nodes in an unweighted graph. To simplify the explanation of our algorithm, we adopt some aspects of Dijkstra's algorithm \citep{Dijkstra1959}, which extends the BFS procedure to find shortest paths in weighted graphs. Instead of checking if a node has been reached, we record the shortest distance from a source node $s$ to all other vertices in the graph. When we extend a path, we compare the length of the previously calculated shortest path from $s$ to $v$ and a new path from $s$ to $u$ and then to $v$ and take the shorter path. Then, to obtain the most degree-central shortest path, we modify the algorithm in the following manner: for each vertex $v$, we place all nodes that can directly precede $v$ on the shortest path from $s$ to $v$ in the set $\mathcal{P}_v$. Each time we consider extending the shortest path $\langle s, \dots, u\rangle$ to $v$, we compare the neighbourhood of the existing most degree-central shortest path from $s$ to $v$ against the neighbourhoods of paths obtained by extending the most degree-central shortest path to $w$, $P_w$, to $u$ then $v$ for all $w \in \mathcal{P}_u$. Furthermore, instead of only considering the case where a path is shorter than a previous path in the traditional Dijkstra's algorithm, we also consider tied shortest paths as we need to compute and compare the degree centrality to ensure the most degree-central shortest path is found.

This is summarised in \Cref{alg:algorithm}.

\begin{algorithm}
    \caption{Finding the most degree-central shortest path}\label{alg:algorithm}
    \begin{algorithmic}[1]
        \Input Graph $G = (V, E)$, starting vertex $s$.
        \Output Most degree-central shortest path $P_v$ from $s$ to $v$ for every $v \in V$.
        \For{$v \in V$}
            \State $d_v \gets + \infty$, $P_v \gets \text{undefined}$, $\mathcal{P}_v = \emptyset$.
        \EndFor
        \State $P_s \gets \langle s \rangle$.
        \State Insert $(0, s)$ to the queue $Q$.
        \While{$Q$ not empty}
            \State $(d_u, u) \gets$ pop next element of $Q$.
            \For{$v \in \{v' \in V \colon (u, v') \in E$\}}
            \State $d_{\text{new}} \gets d_u + 1$.
                \If{$d_{\text{new}} = 1$}
                    \State Insert $(d_{\text{new}}, v)$ to the queue $Q$.
                    \State $d_v \gets d_{\text{new}}$, $P_v = \langle s, v\rangle$, $\mathcal{P}_v \gets \{s\}$.
                \ElsIf{$d_{\text{new}} < d_v$}
                    \State Insert $(d_{\text{new}}, v)$ to the queue $Q$.
                    \State $d_v \gets d_{\text{new}}$, $\mathcal{P}_v \gets \mathcal{P}_v \cup \{u\}$.
                    \For{$w \in \mathcal{P}_u$}
                        \If{$C_d(\langle P_w, u, v\rangle) > C_d(P_v)$}
                            \State $P_v \gets \langle P_w, u, v\rangle$.
                        \EndIf
                    \EndFor
                \ElsIf {$d_{\text{new}} = d_v$}
                    \State $\mathcal{P}_v \gets \mathcal{P}_v \cup \{u\}$.
                    \For{$w \in \mathcal{P}_u$}
                        \If{$C_d(\langle P_w, u, v\rangle) > C_d(P_v)$}
                            \State $P_v \gets \langle P_w, u, v\rangle$.
                        \EndIf
                    \EndFor
                \EndIf
            \EndFor
        \EndWhile
    \end{algorithmic}
\end{algorithm}

\begin{proposition}\label{prop:correctness}
    Given a graph $G = (V, E)$ and a source vertex $s \in V$, \Cref{alg:algorithm} returns the most degree-central shortest paths from $s$ to all other vertices in $V$.
\end{proposition}

\begin{proof}
    We show this by induction. 

    \textit{Base case:} The most degree-central shortest path to the source vertex $s$ is just $\langle s\rangle$. Also, it is easy to see that any path connected in the first iteration of the while loop must be the shortest path from $s$ as it is directly connected to the source node. Furthermore, this will be the most degree-central shortest path as there cannot be another path of the same or shorter length to this set of vertices.

     \textit{Inductive step:} Assume that in the $k$th iteration of the while loop, the most degree-central shortest paths to all nodes in $S \subseteq V$ have been found.

     In the $(k + 1)$th iteration, we consider node $u$ where $d_u$ is the smallest distance between the source vertex $s$ and all other previously connected nodes that have not yet been processed. The for-loop in line 8 allows us to branch to all adjacent nodes, while lines 10, 13 and 21 guarantee that only the shortest path will be recorded. The for-loops and if-statements in lines 16-20 and 23-27 then allow us to extend paths $P_w$ to $u$ then $v$ for all $w \in \mathcal{P}_u$, which, by the inductive hypothesis and \Cref{lem:most-central-path}, will return the most degree-central shortest path from $s$ to $v$ that passes through vertex $u$. After considering extending all shortest paths of length $d_u$ in subsequent iterations of the while-loop, and therefore considering all possible $u \in \mathcal{P}_v$, we are guaranteed the most degree-central shortest path from $s$ to $v$.
\end{proof}

\vspace{\baselineskip}

\begin{proposition}\label{prop:complexity-vertices}
    Given a graph $G = (V, E)$, the worst-case running time of \Cref{alg:algorithm} is $O(|E||V|^2)$. 
\end{proposition}

\begin{proof}
\Cref{alg:algorithm} is a modification of the BFS algorithm. The BFS algorithm has a runtime of $O(|V| + |E|)$ as we enumerate all the edges in the graph once in the outer while loop, and we add to the queue $|V|$ times to ensure all vertices are processed.
    
Assuming that $\mathcal{N}(P_w)$ and $\mathcal{N}(P_v)$ have been stored and can be easily accessed at each iteration of the while loop, computing $\mathcal{N}(\langle P_w, u, v\rangle)$ in the inner for-loop for all possible $w$'s takes $O(|V|^2)$ time in the worst case as we need to consider at most $|V|$ $P_w$'s, and computing $C_d(\langle P_w, u, v\rangle)$ takes $O(|V|)$ time, which dominates the running time of the inner for-loop. This operation is completed at most $O(|E|)$ times as there are $|E|$ iterations of the while-loop to process all the edges. As a result, \Cref{alg:algorithm} has runtime of $O(|E||V|^2)$.
\end{proof}

This bound can be tightened. Let $\Delta(G)$ be the maximum degree of graph $G$. Then, the following is true.
\begin{proposition}\label{prop:complexity}
    Given a graph $G = (V, E)$, the worst-case running time of \Cref{alg:algorithm} is $O(|E| |V| \Delta(G))$.
\end{proposition}
\begin{proof}
    There are only $|\mathcal{P}_u| \leq \Delta(G)$ iterations of the inner for-loops in lines 16-20 and 23-27 as there are at most $\Delta(G)$ nodes that may directly precede vertex $u$ on a shortest path to it. As a result, one order of $|V|$ can be replaced with $\Delta(G)$ from the runtime quoted in \Cref{prop:complexity-vertices}.
\end{proof}

\begin{theorem}\label{thm:runtime}
    Problem~\eqref{eq:degree-problem} can be solved in $O(|E| |V|^2 \Delta(G))$ time.
\end{theorem}
\begin{proof}
    By \Cref{prop:complexity}, \Cref{alg:algorithm} can be used to solve the problem of finding all most degree-central shortest paths from a given source node in $O(|E| |V| \Delta(G))$ time. To solve Problem~\eqref{eq:degree-problem}, we can apply \Cref{alg:algorithm} over all vertices in the network, resulting in an overall running time of $O(|E| |V|^2 \Delta(G))$.
\end{proof}

Comparing this to the running time of Algorithm MVP, which has a runtime of $O(k|V|^6) \sim O(|V|^7)$, with $k$ being the diameter of the graph, \Cref{alg:algorithm} has a runtime of $O(|E| |V|^2 \Delta(G)) \sim O(|V|^5)$ as $|E| \leq |V|^2$ and $\Delta(G) \leq |V|$, which is an improvement of $O(|V|^2)$. In addition, we only require $O(|V|^2)$ space, rather than $O(|V|^3)$.

\section{Degree centrality on weighted graphs}\label{sec:weighted}
In this section, we examine solving Problem~\eqref{eq:degree-problem} on a graph with weighted edges. Hence, throughout this section, we assume that graph $G$ is weighted.

We first show that Problem~\eqref{eq:degree-problem} is NP-hard even with just two distinct weights by showing that it is at least as hard as the \textsc{Maximum Satisfiability} (\textsc{MaxSAT}) problem, which is known to be NP-complete \citep{Krentel1988}.

\vspace{\baselineskip}

\noindent
\begin{tabular*}{\linewidth}{lp{\linewidth-20ex}@{}}
    \toprule
    \multicolumn{2}{l}{\textbf{\textsc{MaxSAT}}} \\
    \textsc{Instance:} & A conjunctive normal form  (CNF) formula consisting of a set of $U$ variables, collection $C$ of clauses over $U$, and positive integer $k$. \\
    \textsc{Question:} & Is there a truth assignment for $U$ that simultaneously satisfies at least $k$ of the clauses in $C$? \\ \bottomrule
\end{tabular*}

\vspace{\baselineskip}
Before we show the reduction, we formally define the decision variant of Problem~\eqref{eq:degree-problem} for a weighted graph.
\vspace{\baselineskip}

\noindent
\begin{tabular*}{\linewidth}{lp{\linewidth-20ex}@{}}
    \toprule
    \multicolumn{2}{l}{\textbf{\textsc{MDCWSP-D}}} \\
    \textsc{Instance:} &  A weighted graph $G$ and a positive integer $\alpha$. \\
    \textsc{Question:} & Is there a shortest path between two vertices in $G$ such that its degree centrality is greater than $\alpha$? \\ \bottomrule
\end{tabular*}

\bigskip

\begin{theorem}\label{thm:np-complete}
    \textsc{MDCWSP-D} is NP-complete.
\end{theorem}
\begin{proof}
First, observe that the degree centrality of a path can be computed in polynomial time; hence, \textsc{MDCWSP-D} is in NP.

Given an instance of \textsc{MaxSAT}, we can construct a directed weighted graph as follows. First, consider an unweighted graph $H$ with $2(|U| + 2)$ nodes and $4(|U| + 1)$ edges with the following structure: we start with two unlinked paths $\langle x_1, \dots, x_{|U|}\rangle$ and $\langle \overline{x}_1, \dots, \overline{x}_{|U|} \rangle$ both of length $|U| - 1$. For $i = 1, \dots, |U| - 1$, we add additional edges $(x_i, \overline{x}_{i + 1})$ and $(\overline{x}_i, x_{i + 1})$ to $H$. This means that the shortest path between any nodes $x_{i}$ or $\overline{x}_i$ and $x_{j}$ or $\overline{x}_j$ is of length $j - i$ for $i < j$. Finally, to complete $H$, we add four additional nodes $s$, $\overline{s}$ that are both connected to $x_1$ and $\overline{x}_1$, and $t$ and $\overline{t}$ that are connected to $x_{|U|}$ and $\overline{x}_{|U|}$. Note that we connect the nodes such that the edges point towards $s$, $\overline{s}$, $t$, and $\overline{t}$. To obtain the weighted graph $G$, we add vertices and edges to $H$ depending on the clauses. For each clause $c \in C$, we add a vertex $y_c$ to $G$ and undirected edges $(z, y_c)$ for each variable $z$ in clause $c$, each with edge weight $|U|$. Given these edge weights, any shortest path in $G$ cannot contain a vertex $y_c$ in the middle of the path because an alternate shorter path will exist. Note: The additional edges can be directed; however, we set them to be undirected to amplify the difference between the unweighted and weighted settings.
    
An example of the structure of $G$ is given in \Cref{fig:np-hardness-proof}.
    
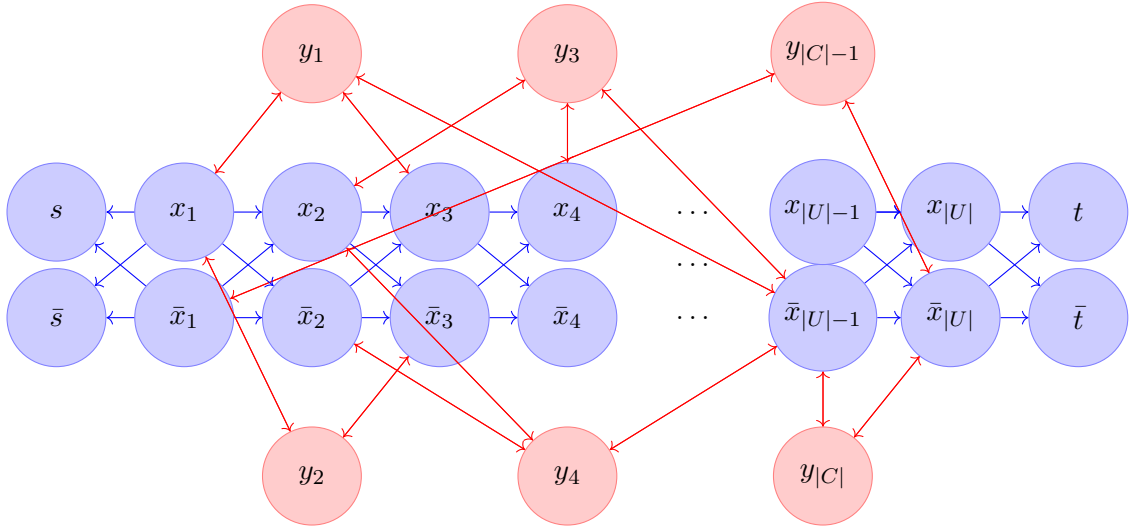
\begin{figure}[!htbt]
\centering
\begin{tikzpicture}[scale=1.4, square/.style={regular polygon, regular polygon sides=4}]
    \begin{scope}[circle, minimum size=13mm]
        \draw
        (0, 0.5) node[draw=blue!50, fill=blue!20] (s){$s$}
        (0, -0.5) node[draw=blue!50, fill=blue!20] (-s){$\overline{s}$}
        (1.2, -0.5) node[draw=blue!50, fill=blue!20] (-1){$\overline{x}_{1}$}
        (1.2, 0.5) node[draw=blue!50, fill=blue!20] (1){$x_{1}$}
        (2.4, -0.5) node[draw=blue!50, fill=blue!20] (-2){$\overline{x}_{2}$}
        (2.4, 0.5) node[draw=blue!50, fill=blue!20] (2){$x_{2}$}
        (3.6, -0.5) node[draw=blue!50, fill=blue!20] (-3){$\overline{x}_{3}$}
        (3.6, 0.5) node[draw=blue!50, fill=blue!20] (3){$x_{3}$}
        (4.8, -0.5) node[draw=blue!50, fill=blue!20] (-4){$\overline{x}_{4}$}
        (4.8, 0.5) node[draw=blue!50, fill=blue!20] (4){$x_{4}$}
        (7.2, -0.5) node[draw=blue!50, fill=blue!20] (-6){$\overline{x}_{|U|-1}$}
        (7.2, 0.5) node[draw=blue!50, fill=blue!20] (6){$x_{|U|-1}$}
        (8.4, -0.5) node[draw=blue!50, fill=blue!20] (-7){$\overline{x}_{|U|}$}
        (8.4, 0.5) node[draw=blue!50, fill=blue!20] (7){$x_{|U|}$}
        (9.6, 0.5) node[draw=blue!50, fill=blue!20] (t){$t$}
        (9.6, -0.5) node[draw=blue!50, fill=blue!20] (-t){$\overline{t}$}
        (2.4, 2) node[draw=red!50, fill=red!20, square, inner sep=-0.2em] (1000){$y_{1}$}
        (2.4, -2) node[draw=red!50, fill=red!20, square, inner sep=-0.2em] (1001){$y_{2}$}
        (4.8, 2) node[draw=red!50, fill=red!20, square, inner sep=-0.2em] (1002){$y_{3}$}
        (4.8, -2) node[draw=red!50, fill=red!20, square, inner sep=-0.2em] (1003){$y_{4}$}
        (7.2, 2) node[draw=red!50, fill=red!20, square, inner sep=-0.2em] (1004){$y_{|C|-1}$}
        (7.2, -2) node[draw=red!50, fill=red!20, square, inner sep=-0.2em] (1005){$y_{|C|}$};
    \end{scope}
    
    \node at ($(4)!.5!(6)$) {\ldots};
    \node at ($(-4)!.5!(-6)$) {\ldots};
    \node at ($(4)!.5!(-6)$) {\ldots};
    
    \begin{scope}[dashed, Stealth-Stealth, thick]
        \draw[red] (-1) to (1004);
        \draw[red] (-2) to (1003);
        \draw[red] (-3) to (1001);
        \draw[red] (-6) to (1002);
        \draw[red] (-6) to (1005);
        \draw[red] (-7) to (1004);
        \draw[red] (-7) to (1005);
        \draw[red] (1) to (1000);
        \draw[red] (1) to (1001);
        \draw[red] (2) to (1003);
        \draw[red] (3) to (1000);
        \draw[red] (-6) to (1000);
        \draw[red] (-6) to (1003);
        \draw[red] (4) to (1002);
        \draw[red] (2) to (1002);
    \end{scope}
    \begin{scope}[-Stealth, thick]
        \draw[blue] (-1) to (2);
        \draw[blue] (-1) to (-2);
        \draw[blue] (1) to (-2);
        \draw[blue] (-2) to (3);
        \draw[blue] (-2) to (-3);
        \draw[blue] (2) to (-3);
        \draw[blue] (-3) to (4);
        \draw[blue] (-3) to (-4);
        \draw[blue] (3) to (-4);
        \draw[blue] (-6) to (7);
        \draw[blue] (-6) to (-7);
        \draw[blue] (6) to (-7);
        \draw[blue] (1) to (2);
        \draw[blue] (2) to (3);
        \draw[blue] (3) to (4);
        \draw[blue] (6) to (7);
        \draw[blue] (6) to (7);
        \draw[blue] (7) to (t);
        \draw[blue] (-7) to (t);
        \draw[blue] (7) to (-t);
        \draw[blue] (-7) to (-t);
        \draw[blue] (1) to (s);
        \draw[blue] (-1) to (s);
        \draw[blue] (1) to (-s);
        \draw[blue] (-1) to (-s);
    \end{scope}
\end{tikzpicture}
\caption{Example of a weighted graph $G$ used in proof of \Cref{thm:np-complete}. The blue circle vertices and edges indicate the initial unweighted graph $H$, and the red square vertices and dashed edges indicate the additional vertices and edges with edge weight $|U|$.}
\label{fig:np-hardness-proof}
\end{figure}

    We now show that there exists a truth assignment for $U$ with more than $k$ satisfiable clauses if and only if there exists a shortest path in $G$ with a degree centrality of at least $|U| + k + 4$.

    $\Longrightarrow$ If there is some truth assignment for $U$ with at least $k$ satisfiable clauses, then the corresponding path in $G$ will have at least $|U| + k + 4$ neighbours by construction as we have $|U|$ neighbours from the \textsc{False} assignments and $k$ neighbours from the satisfied clauses, with an additional $4$ neighbours coming from the source and sink vertices.

    $\Longleftarrow$ Suppose we are given a shortest path in $G$ with at least $|U| + k + 4$ neighbours. We can transform this path into the form $\langle (x_1 \oplus \overline{x}_1), \dots, (x_{|U|} \oplus \overline{x}_{|U|})\rangle\>$, where $\oplus$ is the exclusive or operator, and maintain a degree centrality of at least $|U| + k + 4$. Consider the following cases:
    \begin{itemize}
        \item If the given path ends at $s$, $\overline{s}$, $t$ or $\overline{t}$, we remove this node, and the degree centrality will increase by one.
        \item If the given path starts and ends in $H$, for example, $\langle x_i, \dots, x_j\rangle$, then any arbitrary extension of the path so that it starts at $(x_1 \oplus \overline{x}_1)$ and ends at $(x_{|U|} \oplus \overline{x}_{|U|})$ will maintain or increase the degree centrality as the neighbourhood will contain the unselected nodes in $H$.
        \item If the path starts or ends at one of the additional nodes $y_c$, for example, $\langle y_c, x_i, \dots, x_j\rangle$, we can remove $y_c$ from the path and extend the path using the result from the above case, resulting in the path $\langle (x_{1} \oplus \overline{x}_{1}), \dots, (x_{i-1} \oplus \overline{x}_{i-1}), x_i, \dots, x_j, (x_{j + 1} \oplus \overline{x}_{j + 1}), \dots, (x_{|U|} \oplus \overline{x}_{|U|})\rangle$. This new path will not have a degree centrality smaller than the original path as we will gain $y_c$ as a neighbour and the unselected nodes in $H$.
    \end{itemize}

    In all the above cases, the new paths can be computed in polynomial time. The path returned will be a solution to \textsc{MaxSAT} because after the removal of the two start and end nodes, as well as the unvisited $|U|$ nodes in $H$, the remaining $k$ nodes in the neighbourhood correspond to the satisfied clauses and the path corresponds to the truth assignment.
\end{proof}

\begin{corollary}
    Problem~\eqref{eq:degree-problem} with a weighted graph is NP-hard.    
\end{corollary}

We now discuss two special cases of Problem~\eqref{eq:degree-problem} on a weighted graph and show how to solve them. The first is the case of positive integer-valued weights, and the second is the case where all weights are generated from some positive continuous distribution. 

\textbf{Case 1:} Although Dijkstra's algorithm can be used to solve the problem of finding the shortest path in a graph with positive weighted edges, we note that \Cref{lem:most-central-path} does not hold for weighted graphs. However, if the weights are positive integers, we can augment the network with auxiliary vertices and edges to represent the weighted edges. Then, \Cref{alg:algorithm} can be applied to the augmented graph with a slight modification to ensure that the auxiliary nodes are not counted in the neighbourhood. We demonstrate the importance of this using the example in \Cref{fig:weighted-comparison}. Given the weighted graph in \Cref{fig:original-weighted}, the most degree-central shortest path is $P^\star = \langle 1, 2\rangle$ with $\mathcal{N}(P^\star) = \{3, 4, 5, 6\}$ and $C_d(P^\star) = 4$. We now construct the augmented graph as shown in \Cref{fig:augmented-unweighted}. We label the auxiliary vertices using tuples $(u, v, i)$ to indicate that the new node is $i$ edges away from node $u$ on the original $(u, v)$ arc. If we were to na\"ively apply \Cref{alg:algorithm} on the augmented graph, the optimal path would be $\widetilde{P} = \langle 1, 2, (2, 6, 1), (2, 6, 2), 6 \rangle$ with $\mathcal{N}(\widetilde{P}) = \{(1, 3, 1), (1, 4, 1), (3, 6, 1), 4, 5\}$ and $C_d(\widetilde{P}) = 5$. However, the corresponding path without the auxiliary nodes $\langle 1, 2, 6 \rangle$ only has degree centrality of 3. This discrepancy arises because the auxiliary vertices introduced for the edges linking nodes 3 and 4 to the path should only contribute one neighbour to the neighbourhood.

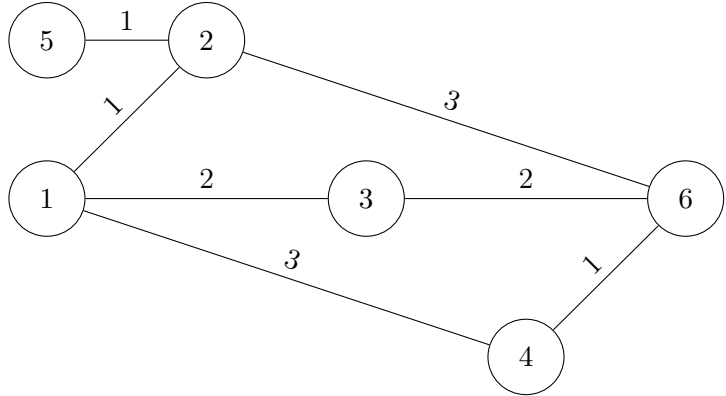
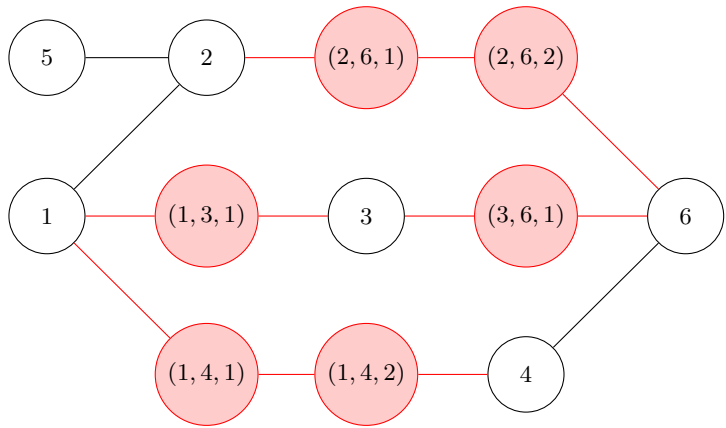
\begin{figure}[!htbt]
    \centering
    \begin{subfigure}[b]{\textwidth}
        \centering
        \begin{tikzpicture}[scale=1.4]
            \clip (-0.5,-2) rectangle (6.5,2);
    
            \begin{scope}[circle, minimum size=10mm]
                \node [draw=black] at (0, 0) (1) {$1$};
                \node [draw=black] at (1.5, 1.5) (2) {$2$};
                \node [draw=black] at (3, 0) (3) {$3$};
                \node [draw=black] at (4.5, -1.5) (4) {$4$};
                \node [draw=black] at (0, 1.5) (5) {$5$};
                \node [draw=black] at (6, 0) (6) {$6$};
            \end{scope}
            
            \path [every node/.style={sloped}]
            (1) edge [black] node [above] {1} (2)
            (1) edge [black] node [above] {2} (3)
            (1) edge [black] node [above] {3} (4)
            (2) edge [black] node [above] {1} (5)
            (2) edge [black] node [above] {3} (6)
            (3) edge [black] node [above] {2} (6)
            (4) edge [black] node [above] {1} (6);
            
        \end{tikzpicture}
        \caption{Original weighted graph (weights on edges).}
        \label{fig:original-weighted}
    \end{subfigure}
    
    \begin{subfigure}[b]{\textwidth}
        \centering
        \begin{tikzpicture}[scale=1.4]
            \clip (-0.5,-2.5) rectangle (6.5,2.5);
    
            \begin{scope}[circle, minimum size=10mm]
                \footnotesize
                \node [draw=black] at (0, 0) (1) {$1$};
                \node [draw=black] at (1.5, 1.5) (2) {$2$};
                \node [draw=black] at (3, 0) (3) {$3$};
                \node [draw=black] at (4.5, -1.5) (4) {$4$};
                \node [draw=black] at (0, 1.5) (5) {$5$};
                \node [draw=black] at (6, 0) (6) {$6$};
                \node [draw=red, fill=red!20] at (1.5, 0) (131)  {$(1,3,1)$};
                \node [draw=red, fill=red!20] at (3, 1.5) (261)  {$(2,6,1)$};
                \node [draw=red, fill=red!20] at (4.5, 1.5) (262)  {$(2,6,2)$};
                \node [draw=red, fill=red!20] at (4.5, 0) (361)  {$(3,6,1)$};
                \node [draw=red, fill=red!20] at (1.5, -1.5) (141) {$(1,4,1)$};
                \node [draw=red, fill=red!20] at (3, -1.5) (142) {$(1,4,2)$};
            \end{scope}
            
            \path [every node/.style={font=\sffamily\small, sloped}]
            (1) edge [black] node [above] {} (2)
            (1) edge [red, dashed] node [above] {} (131)
            (131) edge [red, dashed] node [above] {} (3)
            (1) edge [red, dashed] node [above] {} (141)            
            (141) edge [red, dashed] node [above] {} (142)
            (142) edge [red, dashed] node [above] {} (4)
            (2) edge [black] node [above] {} (5)
            (2) edge [red, dashed] node [above] {} (261)
            (261) edge [red, dashed] node [above] {} (262)
            (262) edge [red, dashed] node [above] {} (6)
            (3) edge [red, dashed] node [above] {} (361)
            (361) edge [red, dashed] node [above] {} (6)
            (4) edge [black] node [above] {} (6);
            
        \end{tikzpicture}
        \caption{Augmented graph (auxiliary nodes and edges in red).} 
        \label{fig:augmented-unweighted}
    \end{subfigure}
    
    \caption{Augmenting a weighted graph.}
    \label{fig:weighted-comparison}
\end{figure}

To avoid the issue of double-counting nodes when computing the neighbourhood of a path, additional work needs to be done to consider the topology of the original graph. This will not worsen the asymptotic running time of \Cref{alg:algorithm} as the additional work in lines 17 and 24 of the algorithm is checking what type of node we are processing. As a result, \Cref{alg:algorithm} will have a running time of $O(w_{\text{sum}} |V|^2 \Delta(G))$, where $w_{\text{sum}}$ is the sum of all edge weights. Note that the running time is not $O(w_{\text{sum}}^3 \Delta(G))$ as we only need to change the number of times the while-loop is executed. The remaining runtime remains the same as we only need to consider setting the starting vertex to be one that is in the original graph, and we only need to compute and compare the degree centrality of the paths when processing a vertex in the original graph. By applying \Cref{alg:algorithm} to the augmented graph, the runtime is now pseudo-polynomial, as it depends on the sum of the edge weights. 

\textbf{Case 2:} If the edge weights $w_{uv}$ are generated from some positive continuous distribution $W_{uv}$ for all $(u, v) \in E$, the shortest path will be unique with probability 1 as $P(W_{uv} = w_{uv}) = 0$ for all $(u, v) \in E$. In this case, the problem can be solved by first solving the all-pairs shortest path problem and then evaluating and comparing the degree centralities of the paths. This can be done in $O(|V|^3)$ time with the Floyd-Warshall algorithm.

\section{Betweenness centrality}\label{sec:btw}
In this section, we consider the problem of finding the most betweenness-central shortest path, which we define as follows.

For a path $P$, we define \emph{path betweenness centrality} $C_b(P)$ to be the number of shortest paths between all pairs of nodes not on $P$ that traverse through at least one node in $P$, i.e.,
\begin{equation*}
C_b(P) := \sum_{u < v \colon u, v \in V \setminus P} g_{uv}(P),    
\end{equation*}
where $g_{uv}(P)$ is the number of shortest (geodesic) paths between nodes $u$ and $v$ that traverse through at least one node in $P$. This measure is sometimes called \emph{stress centrality} \citep{shimbel1953, brandes2001faster}. Note that this definition is different from the one provided by \citet{Everett1999}, who define betweenness centrality for a set of nodes $P$ as
\begin{equation*}
\sum_{u < v \colon u, v \in V \setminus P} \frac{g_{uv}(P)}{g_{uv}},    
\end{equation*}
where $g_{uv}$ is the number of geodesics between $u$ and $v$. We disregard the term in the denominator so that $C_b(P)$ becomes a counting function, which is much easier to work with. If a normalised measure is desired, one may divide $C_b(P)$ by $\sum_{u < v \colon u, v \in V \setminus P}g_{uv}$ to yield a similar result. Furthermore, our definition of path betweenness centrality differs from the definition provided by \citet{PuzisRami2007}, who only count the geodesics that traverse all nodes on $P$. We prefer our definition because it is less restrictive and, as a result, is better suited for the applications that motivate our work. For example, it can be used to count the number of direct paths that an adversarial player is able to compromise in all-to-all communication networks \citep{Brandes2005NetworkAnalysis}. In this setting, we do not require communication to pass through all infected machines because any message that passes through at least one infected computer will be compromised.

It then follows that the problem of finding the most betweenness-central shortest path becomes
\begin{equation}\label{eq:betweenness-path-problem}
    \max \left\{C_b(P) \colon P \in \mathcal{SP}(G)\right\}.
\end{equation}

To solve this problem, we rely on the following lemma:
\begin{lemma}\label{lem:most-betweenness-shortest-path}
    If $\langle s, p_1, \ldots, p_{k}, t \rangle$ is a most betweenness-central shortest path from $s$ to $t$, then $\langle s, p_1, \ldots, p_{k} \rangle$ is a most betweenness-central shortest path from $s$ to $p_{k}$.
\end{lemma}

\begin{proof}
    By \Cref{lem:shortest-path}, it is clear that if $\langle s, p_1, \ldots, p_{k}, t \rangle$ is a shortest path from $s$ to $t$ then $\langle s, p_1, \ldots, p_{k} \rangle$ is a shortest path from $s$ to $p_k$.

Consider the graph in \Cref{fig:most-betweenness-central-structure-end-t}. Let $\langle s, p_1, \ldots p_{k-1}, p_{k}, t\rangle$ be the most betweenness-central shortest path between $s$ and $t$ and $\langle s, \tilde{p}_1, \ldots, \tilde{p}_{k - 1}, p_k, t\rangle$ be an alternative path.

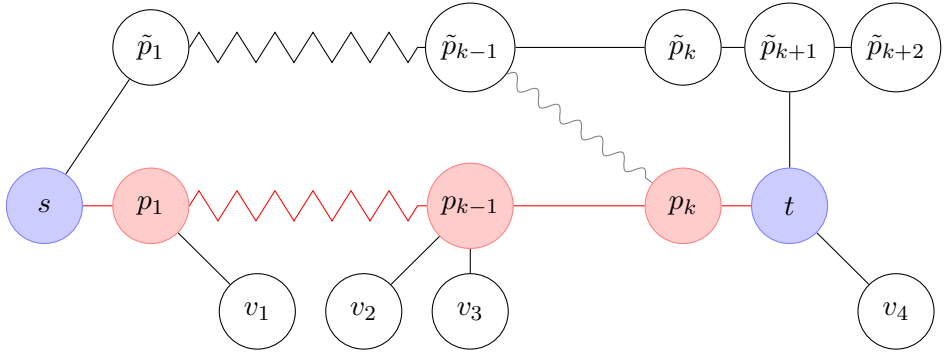
\begin{figure}[H]
    \centering
    \begin{tikzpicture}[scale=1.4]
        \clip (-0.5,-1.8) rectangle (8.5,1.8);
    
        \begin{scope}[circle,minimum size=10mm]
            \node at (0, 0) [draw=blue!50,fill=blue!20] (s) {$s$};
            \node at (8, 0) [draw=blue!50,fill=blue!20] (t) {$t$};
    
            \node at (6, 0) [draw=red!50,fill=red!20] (k) {$p_{k}$};
            
            \node at (1, 1) [draw=black] (1) {$p_1$};
            \node at (3, 1) [draw=black] (j) {$p_{j}$};
            \node at (5, 1) [draw=black] (k-1) {$p_{k - 1}$};
            
            \node at (1, -1) [draw=black] (til1) {$\tilde{p}_1$};
            \node at (3, -1) [draw=black] (tilj) {$\tilde{p}_j$};
            \node at (5, -1) [draw=black] (tilk-1) {$\tilde{p}_{k - 1}$};
        \end{scope}
        
        \path [every node/.style={font=\sffamily\small, sloped}]
        (s) edge [black] node [above] {} (1)
        (k-1) edge [black] node [above] {} (k)
        (s) edge [black] node [above] {} (til1)
        (tilk-1) edge [black] node [above] {} (k)
        (k) edge [black] node [above] {} (t)
        (j) edge [gray, bend left=40] node [above] {} (t)
        (tilj) edge [gray, bend right=40] node [above] {} (t);

        \draw [decorate, decoration={zigzag, segment length=5mm, amplitude=2mm}, black] (1) -- (j);
        \draw [decorate, decoration={zigzag, segment length=5mm, amplitude=2mm}, black] (j) -- (k-1);
        \draw [decorate, decoration={zigzag, segment length=5mm, amplitude=2mm}, black] (til1) -- (tilj);
        \draw [decorate, decoration={zigzag, segment length=5mm, amplitude=2mm}, black] (tilj) -- (tilk-1);

    \end{tikzpicture}
    \caption{Illustration of graph used in the proof of \Cref{lem:most-betweenness-shortest-path}.}
    \label{fig:most-betweenness-central-structure-end-t}
\end{figure}

Given that $\langle s, p_1, \ldots, p_{k}, t \rangle$ is the most betweenness-central shortest path between $s$ and $t$, removing $t$ from the end of the path will update the betweenness centrality score by adding paths that end at $t$ and traverse through at least one of the nodes in $\langle s, p_1, \ldots, p_{k} \rangle$ and subtracting paths that only traverse through $t$ and none of the other nodes in the remaining path. The number of paths subtracted in this case is fixed regardless of whether $\langle s, p_1, \ldots, p_{k}, t \rangle$ or $\langle s, \tilde{p}_1, \ldots, \tilde{p}_{k - 1}, p_k, t\rangle$ was the most betweenness-central shortest path between $s$ and $t$. This, in turn, implies that for $\langle s, p_1, \ldots, p_{k-1}, p_k, t \rangle$ to be the most betweenness-central shortest path, $\langle s, p_1, \ldots, p_{k-1}, p_k \rangle$ must be the most betweenness-central shortest path between $s$ and $p_k$ as this path already accounts for shortest paths that end in $t$ that traverse through at least one node in the path.
\end{proof}

\Cref{lem:most-betweenness-shortest-path} implies that an algorithm similar to \Cref{alg:algorithm} can be applied to solve Problem~\eqref{eq:betweenness-path-problem}; however, some simplifications can be made. Comparing \Cref{lem:most-central-path} and \Cref{lem:most-betweenness-shortest-path}, we observe that we only need to extend the existing most betweenness-central shortest paths by one node rather than two nodes. This means that the nested for-loops in lines 16-20 and 23-27 can be removed and replaced with an if-statement that checks whether $C_b(\langle P_u, v\rangle)$ is larger than $C_b(P_v)$. For clarity, we provide the pseudo-code of this procedure in \Cref{alg:betweenness-algorithm}.

\begin{algorithm}
    \caption{Finding the most betweenness-central shortest path}\label{alg:betweenness-algorithm}
    \begin{algorithmic}[1]
        \Input Graph $G = (V, E)$, starting vertex $s$.
        \Output Most betweenness-central shortest path $P_v$ from $s$ to $v$ for every $v \in V$.
        \For{$v \in V$}
            \State $d_v \gets + \infty$, $P_v \gets \text{undefined}$.
        \EndFor
        \State $P_s \gets \langle s \rangle$.
        \State Insert $(0, s)$ to the queue $Q$.
        \While{$Q$ not empty}
            \State $(d_u, u) \gets$ pop next element of $Q$.
            \For{$v \in \{v' \in V \colon (u, v') \in E$\}}
            \State $d_{\text{new}} \gets d_u + 1$.
                \If{$d_{\text{new}} < d_v$}
                    \State Insert $(d_{\text{new}}, v)$ to the queue $Q$.
                    \State $d_v \gets d_{\text{new}}$, $P_v = \langle s, v\rangle$.
                \ElsIf{$d_{\text{new}} = d_v$}
                    \If{$C_b(\langle P_u, v\rangle) > C_b(P_v)$}
                        \State $P_v \gets \langle P_u, v\rangle$.
                    \EndIf
                \EndIf
            \EndFor
        \EndWhile
    \end{algorithmic}
\end{algorithm}

\begin{lemma}
    \Cref{alg:betweenness-algorithm} has a running time of $O(|E|^2|V|)$
\end{lemma}

\begin{proof}
    \Cref{alg:betweenness-algorithm} is a simplification of \Cref{alg:algorithm}, with the removal of the innermost for-loop and replacement of the centrality measure. As with \Cref{alg:algorithm}, the amount of work required is dominated by the inner for-loop, which is executed a total of $|E|$ times. Within this for-loop, the betweenness of the path needs to be calculated, taking $O(|E||V|)$ time \citep{Brandes2008}. As a result, \Cref{alg:betweenness-algorithm} has a running time of $O(|E|^2|V|)$.
\end{proof}

\begin{remark}\label{remark}
    As noted by \citet{PuzisRami2007} and \citet{Brandes2008}, when computing group betweenness centrality for many different groups, the preprocessing stage of the algorithm (named initialisation, path discovery and counting in \citet{Brandes2008}) only needs to be completed once on each graph and the results can be stored and used to compute the group betweenness centrality in the accumulation stage when required. This does not improve the asymptotic running time of the algorithm as we still require $O(|E||V|)$ time to compute the betweenness centrality of a group after spending $O(|E||V|)$ time on preprocessing.
\end{remark}

\begin{theorem}
    Problem~\eqref{eq:betweenness-path-problem} is solvable in $O(|E|^2|V|^2)$ time.
\end{theorem}
\begin{proof}
    In order to solve Problem~\eqref{eq:betweenness-path-problem}, \Cref{alg:betweenness-algorithm} is executed $|V|$ times, once for each vertex in the network, resulting in an overall runtime of $O(|E|^2|V|^2)$.
\end{proof}

Furthermore, solving Problem~\eqref{eq:betweenness-path-problem} with positively weighted graphs is easier than solving Problem~\eqref{eq:degree-problem}. By \Cref{lem:most-betweenness-shortest-path}, we only need to look back to the preceding node. As a result, \Cref{lem:most-betweenness-shortest-path} also holds for graphs with positively weighted edges. This means that \Cref{alg:betweenness-algorithm} can be adapted to handle graphs with positively weighted edges by updating the new distance with $d_{\text{new}} \gets d_u + w_{uv}$ in line 9 and using a priority queue instead of a standard queue, which is more akin to Dijkstra's algorithm \citep{Dijkstra1959} than the standard BFS procedure. This leads to the following result.
\begin{corollary}
    Problem~\eqref{eq:betweenness-path-problem} on a graph with positively weighted edges is solvable in polynomial time.
\end{corollary}
\begin{proof}
    For a given source node $s$, \Cref{alg:betweenness-algorithm} can be applied but with a priority queue instead of a simple queue. The running time will still be $O(|E|^2|V|)$ as the innermost for-loop is dominated by the cost of computing the path betweenness centrality in $O(|E||V|)$ time rather than the $O(\log(|V|))$ time taken to update the priority queue. As \Cref{alg:betweenness-algorithm} needs to be applied on all possible starting nodes to solve Problem~\eqref{eq:betweenness-path-problem}, the overall running time is $O(|E|^2|V|^2)$.

    As mentioned above, the preprocessing stage required to compute the group betweenness centrality only needs to be performed once for a graph. This will take $O(|V|^2\log(|V|))$ time \citep{Brandes2008}.

    This implies that the overall time required to solve Problem~\eqref{eq:betweenness-path-problem} is $O(|E|^2|V|^2 + |V|^2\log(|V|)) \sim O(|E|^2|V|^2)$ as $|V| < |E|$.
\end{proof}

\section{Closeness centrality}\label{sec:cls}

We now consider the problem of finding the most closeness-central shortest path, where \emph{path closeness centrality} is defined as follows. 

Let $d(u, P)$ be the shortest distance from a node $u$ to the path $P$,
\[
d(u, P) := \min\{d_{uv} \colon v \in P\},
\]
where $d_{uv}$ is the length of the shortest path from $u$ to $v$. To compute path closeness centrality for $P$, we take the maximum over all nodes in $G$: 
\[
C_c^{\max}(P) := \max\{d(u, P) \colon u \in V\}.
\]
This measure is sometimes called \emph{path eccentricity} \citep{Slater1982}. Observe that, consistent with the classic definitions in \citet{Everett1999}, $C_c^{\max}(P)$ is an inverse measure of closeness as larger values indicate less centrality.  Although other aggregation functions can be used to define path closeness centrality (for example, the sum, see \citet{Vogiatzis2014}), we choose this definition as it represents the farthest distance one must travel to reach the path $P$. In particular, $C_c^{\max}(P)$ could represent the largest number of towns through which any commuter in a region must travel to reach the closest train station. This observation leads to the problem of finding the most closeness-central shortest path:
\begin{align}\label{eq:closest-path-problem-max}
    \max \left\{\frac{1}{C_c^{\max}(P)} \colon P \in \mathcal{SP}(G)\right\}.
\end{align}
Note that in Problem~\eqref{eq:closest-path-problem-max}, we seek to maximise the inverse of $C_c^{\max}(P)$ to be consistent with Problem~\eqref{eq:problem}. 
However, it is more natural to work with the equivalent minimization problem:
\begin{align}\label{eq:closest-path-problem}
    \min \left\{C_c^{\max}(P) \colon P \in \mathcal{SP}(G)\right\}.
\end{align}

Next, we show that Problem~\eqref{eq:closest-path-problem} is as hard as the \textsc{Satisfiability} (\textsc{SAT}) problem, which is known to be NP-complete \citep{GareyJohnson1979}.

\vspace{\baselineskip}

\noindent
\begin{tabular*}{\linewidth}{lp{\linewidth-20ex}@{}}
    \toprule
    \multicolumn{2}{l}{\textbf{\textsc{SAT}}} \\
    \textsc{Instance:} & A conjunctive normal form (CNF) formula consisting of a set of $U$ variables and a collection $C$ of clauses over $U$. \\
    \textsc{Question:} & Is there a satisfying truth assignment for $C$? \\ \bottomrule
\end{tabular*}

\vspace{\baselineskip}
We define the corresponding decision version of Problem~\eqref{eq:closest-path-problem} (\textsc{MCCSP-D}) below.

\vspace{\baselineskip}

\noindent
\begin{tabular*}{\linewidth}{lp{\linewidth-20ex}@{}}
    \toprule
    \multicolumn{2}{l}{\textbf{\textsc{MCCSP-D}}} \\
    \textsc{Instance:} &  A graph $G$ and a positive integer $\alpha$. \\
    \textsc{Question:} & Is there a shortest path between two vertices in $G$ such that its closeness centrality is less than $\alpha$? \\ \bottomrule
\end{tabular*}

\vspace{\baselineskip}

\begin{theorem}\label{thm:closest-shortest-path}
    \textsc{MCCSP-D} is NP-complete.
\end{theorem}
\begin{proof}
    \textsc{MCCSP-D} is in NP because all pairwise distances in $G$ can be computed in polynomial time. Hence, computing the closeness centrality of a path can also be completed in polynomial time.

    Given an instance of \textsc{SAT}, we construct a graph $G$ in the following manner. Similar to the proof of \Cref{thm:np-complete}, we start by constructing a graph $H$ with $2(|U| + 1)$ nodes and $4|U|$ directed edges. Again, we start with two unlinked paths $\langle x_1, \dots, x_{|U|}\rangle$ and $\langle \overline{x}_1, \dots, \overline{x}_{|U|} \rangle$ both of length $|U| - 1$. For $i = 1, \dots, |U| - 1$, we add edges $(x_i, \overline{x}_{i + 1})$ and $(\overline{x}_i, x_{i + 1})$ to connect these two paths. To complete $H$, we add in two nodes, $s$ and $t$, connected to the existing graph with the following edges $(s, x_1)$, $(s, \overline{x}_1)$, $(x_{|U|}, t)$ and $(\overline{x}_{|U|}, t)$. Finally, to obtain $G$, we start with $H$, and for each clause $c \in C$, we add the following nodes and edges. For each variable $j$ in clause $c$, where $j \in \cup_{i=1}^{|U|}\{x_i, \overline{x}_i\}$, we construct an undirected path $\langle j, z_{c, j, 1}, \dots, z_{c, j, |U| - 1}, y_c\rangle$ of length $|U|$. Overall, $G$ has $2(|U| + 1) + |C| + C_{\text{len}}(|U| - 1)$ nodes and $4|U|$ directed and $C_{\text{len}}(|U| - 1)$ undirected edges, where $C_{\text{len}} = \sum_{i=1}^{|C|} |c_i|$ is the length of the formula, i.e., the number of variables in the clauses. Comparing this graph to the one used in the proof of \Cref{thm:np-complete}, there is only one copy of the source and sink nodes $s$ and $t$, and the paths from $H$ to $y_c$ consists of $|U|$ unit-weighted edges rather than a single edge with weight $|U|$. 

    We now show that a solution to \textsc{SAT} exists if and only if there exists a shortest path in $G$ with a closeness centrality less than or equal to $|U|$.
    
    $\Longrightarrow$ Given a solution to \textsc{SAT}, the corresponding path in $G$ will have a closeness centrality of exactly $|U|$ (or 1 if the instance of \textsc{SAT} contains no clauses). This is because any node $y_c$ that is associated with an unsatisfied clause will be $|U| + 1$ units away from the shortest path.
    
    $\Longleftarrow$ Given a solution to \textsc{MCCSP-D} with a closeness centrality less than or equal $|U|$, we have the following cases:
    \begin{enumerate}
        \item The path starts at $s$ and ends at $t$. Then, removing $s$ and $t$ yields the corresponding satisfying truth assignment for $C$.
        \item The path starts at $s$ or $(x_i \oplus \overline{x}_i)$ and ends at $t$ or $(x_j \oplus \overline{x}_j)$. Then, any arbitrary extension of the path through any nodes in $H$ such that the path starts at $s$ and ends at $t$ leads to the first case.
        \item The path starts and/or ends at one of the nodes in $G$ but not in $H$. Then, removing the nodes only in $G$ and arbitrarily extending the path such that it starts at $s$ and ends at $t$ again leads to the first case without increasing the closeness-centrality to be greater than $|U|$. Note that if the path only traverses nodes corresponding to one clause, then the mapping to \textsc{SAT} is trivial.
    \end{enumerate}
    
    As these transformations can be computed in polynomial time, \textsc{MCCSP-D} is NP-complete.
\end{proof}

\begin{corollary}\label{cor:closest-path-problem-np-hard}
    Problem~\eqref{eq:closest-path-problem} is NP-hard.
\end{corollary}
As a result of \Cref{cor:closest-path-problem-np-hard}, a polynomial-time algorithm for Problem~\eqref{eq:closest-path-problem} does not exist unless $\text{P} = \text{NP}$.

\section{Computational experiments}\label{sec:computational-experiments}
We tested \Cref{alg:algorithm} and \Cref{alg:betweenness-algorithm} on a large number of synthetic and real-world graphs. The experiments were conducted on a Mac Studio (2023) with Apple M2 Ultra and 128 GB of RAM running macOS Ventura 13.6. The algorithms were implemented in Python 3.12.

\subsection{Data and results}
The synthetic graphs were randomly generated using two procedures outlined below (each type of graph was generated 30 times for each parameter setting):

\begin{itemize}
    \item \textit{Watts-Strogatz} graphs. We considered the `small-world' network proposed by \citet{WattsStrogatz1998}, which represents networks that are ``highly clustered, like regular lattices, yet have small characteristic path lengths''. We considered Watts-Strogatz graphs with 100, 500, 1000, 5000, and 10000 nodes. The number of initial neighbours was set to 4, and the rewiring probability was set to one of 0.1 or 0.2.
    \item \textit{Barab\'{a}si-Albert} graphs. The `scale-free preferential attachment' model proposed by \citet{BarabasiAlbert1999} generates graphs with edges whose degree follows the power law. We considered graphs with the same number of nodes as the Watts-Strogatz instances and set the number of edges attached to any new node to 2 to mirror the edge density of the Watts-Strogatz graphs.
\end{itemize}
In total, we generated and analysed 450 synthetic instances. We averaged the results for each parameter instance and summarised them in \Crefrange{tab:ws1}{tab:ba}.

To compare our results to the MVP algorithm, we also considered the following undirected real-world instances:
\begin{itemize}    
    \item \textit{Sandi Auths}: An academic collaboration network \citep{RossiAhmed2015}.
    \item \textit{IEEE Bus}: A subnetwork of the US Electric Power System in 1962 \citep{DavisHu2011}.  
    \item \textit{Santa Fe}: A collaboration network for the Santa Fe Institute \citep{Girvan2002}. 
    \item \textit{US Air 97}: A transportation network of US Air in 1997 \citep{BatageljMrvar2006, DavisHu2011}.
    \item \textit{Bus}: Bus power system network \citep{DavisHu2011}.
    \item \textit{Email}: Network of e-mail interchanges \citep{Guimera2003, DavisHu2011}.
    \item \textit{Cerevisiae}: Biological network of yeast protein-protein interactions \citep{DavisHu2011}.
\end{itemize}
The results from these real-world graph instances are presented in \Cref{tab:real}.

In addition to the unweighted graphs, we also consider the following weighted graphs:
\begin{itemize}
    \item \textit{Copenhagen calls} and \textit{Copenhagen SMS}: Undirected and directed graphs with weights for the number of phone calls made and text messages sent between university students in Copenhagen over multiple weeks \citep{SapiezynskiPiotr2019}.
    \item \textit{Bitcoin Alpha} and \textit{Bitcoin OTC}: Directed graphs of trust scores between users on the Bitcoin Alpha and Bitcoin OTC platforms \citep{Kumar2016}. Note that we have transformed the weights from a $[-10, 10]$ range to $[1, 21]$ to ensure positive weights.
    \item \textit{Advogato}: Directed graph of trust scores from users of Advogato, an open source software community \citep{Massa2009}. Similar to the Bitcoin graphs, the weights have been rescaled from having weights of $\{0.6, 0.8, 1.0\}$ to $\{3, 4, 5\}$.
\end{itemize}
For these weighted graphs, we first examine the performance of \Cref{alg:algorithm} and \Cref{alg:betweenness-algorithm} on the unweighted graphs obtained by ignoring the weights. We then consider the weighted graphs by applying \Cref{alg:betweenness-algorithm} to the original graphs and \Cref{alg:algorithm} to the augmented graphs. The results are presented in \Crefrange{tab:weighted-copenhagen}{tab:weighted-trust}. Note that these tables also include rows for the number of nodes traversed so that comparisons can be made to the unweighted case.

In addition to the standard statistics, we report the diameter of the graph (diam), the degree of the graph ($\Delta(G)$), and the relative size of the feasible set $|\mathcal{SP}(G)|/U(G)$, where $U(G)$ is the number of pairwise connectable nodes in $G$ (i.e., the number of shortest paths if they were all unique).

\Cref{alg:algorithm} was examined on all of the above instances, but we restrict our analysis of \Cref{alg:betweenness-algorithm} to graphs with less than 1000 vertices.

\begin{table}[htb]
    \centering
    \footnotesize
    \caption{Results for Watts-Strogatz graphs with rewiring probability set to 0.1, averaged over 30 instances. Runtimes and preprocessing times are in seconds. Rows \textit{path length} and \textit{path centrality} give the length and the centrality of the optimal shortest path, respectively.}
    \begin{tabular}{lrrrrr}
        \toprule
	& \multicolumn{5}{c}{Watts-Strogatz (4, 0.1)} \\
	\cmidrule(l){2-6}
	$|V|$                           & 100      & 500       & 1000   & 5000   & 10000   \\
	$|E|$                           & 200      & 1000      & 2000   & 10000  & 20000   \\
	$\Delta(G)$                     & 5.87     & 6.40      & 6.70   & 7.27   & 7.33    \\
	$|\mathcal{SP}(G)| / U(G)$      & 2.45     & 3.18      & 3.53   & 4.37   & 4.80    \\
	diam                            & 10.43    & 15.40     & 17.73  & 22.33  & 24.50   \\     
    \midrule
    \multicolumn{6}{c}{Degree centrality} \\    
    \midrule
	diam centrality                  & 22.23    & 31.20     & 35.37  & 42.53  & 45.20   \\
	path length                      & 9.03     & 12.97     & 14.73  & 18.63  & 19.37   \\
        path centrality                  & 23.17    & 33.87     & 37.60  & 46.90  & 50.40   \\ 	        
	MVP runtime                      & 0.41     & 59.73     & 513.84 & -      & -       \\
	Alg.~\ref{alg:algorithm} runtime & 0.07     & 2.08      & 10.26  & 366.88 & 1712.52 \\
        \midrule
        \multicolumn{6}{c}{Betweenness centrality} \\
        \midrule
	diam centrality                  & 11823.33 & 172901.27 & -      & -      & -       \\  
	path length                      & 8.27     & 11.73     & -      & -      & -       \\
    path centrality                  & 13130.13 & 213512.87 & -      & -      & -       \\	
	preprocessing time               & 0.01     & 0.18      & -      & -      & -       \\
	Alg.~\ref{alg:betweenness-algorithm} runtime                  & 38.61    & 28597.39  & -      & -      & -       \\	
    \bottomrule
    \end{tabular}
    \label{tab:ws1}
\end{table}

\begin{table}[htb]
    \centering
    \footnotesize
    \caption{Results for Watts-Strogatz graphs with rewiring probability set to 0.2, averaged over 30 instances. Runtimes and preprocessing times are in seconds. Rows \textit{path length} and \textit{path centrality} give the length and the centrality of the optimal shortest path, respectively.}
    \begin{tabular}{lrrrrr}
	\toprule
	 &  \multicolumn{5}{c}{Watts-Strogatz (4, 0.2)} \\
	\cmidrule(l){2-6}
	$|V|$                           & 100     & 500       & 1000   & 5000   & 10000   \\
	$|E|$                           & 200     & 1000      & 2000   & 10000  & 20000   \\
	$\Delta(G)$                     & 6.60    & 7.40      & 7.43   & 8.03   & 8.20    \\
	$|\mathcal{SP}(G)| / U(G)$      & 2.04    & 2.25      & 2.35   & 2.57   & 2.66    \\
	diam                            & 8.43    & 11.50     & 12.97  & 16.13  & 17.47   \\ 
        \midrule
        \multicolumn{6}{c}{Degree centrality} \\    
        \midrule
	diam centrality                 & 21.67   & 28.80     & 30.60  & 37.73  & 40.17   \\
	path length                     & 7.03    & 9.33      & 10.17  & 12.53  & 13.27   \\
	path centrality                 & 23.33   & 31.73     & 35.17  & 43.77  & 47.07   \\        
	MVP runtime                     & 0.41    & 55.73     & 482.51 & -      & -       \\
	Alg.~\ref{alg:algorithm} runtime & 0.06    & 2.08      & 9.75   & 359.95 & 1704.02 \\ 
        \midrule
        \multicolumn{6}{c}{Betweenness centrality} \\
        \midrule
	diam centrality                 & 7343.40 & 81079.20  & -      & -      & -       \\
	path length                     & 6.83    & 8.83      & -      & -      & -       \\
	path centrality                 & 8283.87 & 100763.80 & -      & -      & -       \\
	preprocessing time              & 0.01    & 0.18      & -      & -      & -       \\
	Alg.~\ref{alg:betweenness-algorithm} runtime                  & 36.89   & 26936.55  & -      & -      & -       \\
    \bottomrule
    \end{tabular}
\label{tab:ws2}
\end{table}

\begin{table}[htb]
    \footnotesize
    \centering
    \caption{Results for Barab\'{a}si-Albert graphs averaged over 30 instances. Runtimes and preprocessing times are in seconds. Rows \textit{path length} and \textit{path centrality} give the length and the centrality of the optimal shortest path, respectively.}
    \begin{tabular}{lrrrrr}
        \toprule
	{} & \multicolumn{5}{c}{Barab\'{a}si-Albert} \\
	\cmidrule(l){2-6}
	$|V|$                           & 100      & 500       & 1000   & 5000   & 10000   \\
	$|E|$                           & 196      & 996       & 1996   & 9996   & 19996   \\
	$\Delta(G)$                     & 25.27    & 53.83     & 76.53  & 188.20 & 257.40  \\
	$|\mathcal{SP}(G)| / U(G)$      & 2.13     & 2.67      & 2.87   & 3.34   & 3.53    \\
	diam                            & 5.57     & 7.03      & 7.27   & 8.53   & 9.00    \\ 
        \midrule
        \multicolumn{6}{c}{Degree centrality} \\    
        \midrule
	diam centrality                 & 46.17    & 120.10    & 179.30 & 382.57 & 553.27  \\
	path length                     & 3.87     & 4.67      & 5.07   & 5.80   & 6.03    \\
	path centrality                 & 51.87    & 138.33    & 199.83 & 483.93 & 676.47  \\
	MVP runtime                     & 0.36     & 45.92     & 381.45 & -      & -       \\
	Alg.~\ref{alg:algorithm} runtime & 0.07     & 2.31      & 13.46  & 713.06 & 3792.77 \\ 
        \midrule
        \multicolumn{6}{c}{Betweenness centrality} \\
        \midrule
	diam centrality                 & 12879.93 & 346136.47 & -      & -      & -       \\
	path length                     & 4.37     & 5.37      & -      & -      & -       \\
	path centrality                 & 14788.67 & 402991.40 & -      & -      & -       \\	
	preprocessing time              & 0.01     & 0.18      & -      & -      & -       \\
	Alg.~\ref{alg:betweenness-algorithm} runtime                  & 34.66    & 23966.26  & -      & -      & -       \\	
    \bottomrule
    \end{tabular}
    \label{tab:ba}
\end{table}

\begin{table*}[htb]
    \centering
    \footnotesize
    \caption{Results for real-world instances. Runtimes and preprocessing times are in seconds. Rows \textit{path length} and \textit{path centrality} give the length and the centrality of the optimal shortest path, respectively.}
    
    \begin{tabular}{lrrrrrrrr}
	\toprule
	                                & Sandi Auths & IEEE Bus & Santa Fe & US Air 97 & Bus       & Email  & Cerevisiae \\
	\cmidrule(l){2-2}\cmidrule(l){3-3}\cmidrule(l){4-4}\cmidrule(l){5-5}\cmidrule(l){6-6}\cmidrule(l){7-7}\cmidrule(l){8-8}
	$|V|$                           & 86          & 118      & 118      & 332       & 662       & 1133   & 1458       \\
	$|E|$                           & 124         & 179      & 200      & 2126      & 906       & 5451   & 1948       \\      
	$\Delta(G)$                     & 12          & 9        & 29       & 139       & 9         & 71     & 56         \\  
	$|\mathcal{SP}(G)| / U(G)$      & 1.36        & 2.26     & 1.51     & 5.55      & 2.44      & 6.73   & 2.57       \\
	diam                            & 11          & 14       & 12       & 6         & 25        & 8      & 19         \\ 
        \midrule
        \multicolumn{8}{c}{Degree centrality} \\    
        \midrule
	diam centrality                 & 32          & 32       & 90       & 167       & 45        & 159    & 57         \\
	path length                     & 7           & 8        & 10       & 3         & 20        & 4      & 7          \\
	path centrality                 & 34          & 33       & 92       & 206       & 50        & 187    & 156        \\
	MVP runtime                     & 0.25        & 0.73     & 0.66     & 21.96     & 147.03    & 807.92 & 1317.34    \\
	Alg.~\ref{alg:algorithm} runtime & 0.04        & 0.09     & 0.09     & 2.17      & 3.88      & 26.53  & 22.89      \\ 
        \midrule
        \multicolumn{8}{c}{Betweenness centrality} \\
        \midrule
	diam centrality                 & 8540        & 26530    & 20422    & 180104    & 650164    & -      & -          \\
	path length                     & 7           & 13       & 11       & 5         & 18        & -      & -          \\
	path centrality            & 8566        & 26734    & 20422    & 254286    & 705878    & -      & -          \\	
	preprocessing time              & 0.00        & 0.01     & 0.01     & 0.15      & 0.29      & -      & -          \\
	Alg.~\ref{alg:betweenness-algorithm} runtime                  & 13.28       & 61.18    & 47.08    & 10612.16  & 76869.96  & -      & -          \\	
    \bottomrule
    \end{tabular}
    \label{tab:real}
\end{table*}

\begin{table*}[htb]
    \centering
    \footnotesize
    \caption{Results for degree centrality on Copenhagen calls and SMS graph instances. Note that $|E| = w_{\text{sum}}$ for the augmented graphs. Runtimes are in seconds. Rows \textit{path length} and \textit{path centrality} give the length and the centrality of the optimal shortest path, respectively.}
    \label{tab:weighted-copenhagen}
    \begin{tabular}{lrrrrrrrr}
        \toprule
                                 & \multicolumn{4}{c}{Copenhagen calls} & \multicolumn{4}{c}{Copenhagen SMS} \\ \cmidrule(l){2-5}\cmidrule(l){6-9}
        weighted                 & no      & yes     & no      & yes    & no      & yes    & no     & yes    \\
        directed                 & yes     & yes     & no      & no     & yes     & yes    & no     & no     \\
        $|V|$                    & 536     & 3212    & 536     & 3515   & 568     & 23598  & 568    & 24204  \\
        $|E|$                    & 924     & 3600    & 621     & 3600   & 1303    & 24333  & 697    & 24333  \\
        $\Delta(G)$              & 18      & 18      & 18      & 18     & 11      & 11     & 11     & 11     \\
        $|\mathcal{SP}(G)|/U(G)$ & 1.43    & 1.35    & 1.79    & 1.50   & 1.99    & 1.30   & 2.12   & 1.32  \\
        diam                     & 21      & 75      & 22      & 197    & 22      & 1783   & 20     & 3781   \\  
        \midrule
        \multicolumn{9}{c}{Degree centrality} \\    
        \midrule
        diam nodes traversed     & 22      & 25      & 23      & 15     & 23      & 10     & 21     & 12     \\
        diam centrality          & 61      & 25      & 41      & 42     & 30      & 22     & 42     & 30     \\
        path length              & 16      & 22      & 15      & 24     & 18      & 61     & 10     & 67     \\
        path nodes traversed     & 17      & 15      & 16      & 15     & 19      & 20     & 11     & 20     \\
        path centrality          & 63      & 52      & 66      & 63     & 49      & 59     & 50     & 64     \\
        Alg.~\ref{alg:algorithm} runtime           & 0.40    & 4.31    & 0.81    & 4.53   & 1.10    & 74.19  & 1.33   & 96.71  \\ 
        \bottomrule
    \end{tabular}
\end{table*}

\begin{table*}[htb]
    \centering
    \footnotesize
    \caption{Results for betweenness centrality on Copenhagen calls and SMS graph instances. Runtimes and preprocessing times are in seconds. Rows \textit{path length} and \textit{path centrality} give the length and the centrality of the optimal shortest path, respectively.}
    \label{tab:weighted-copenhagen-betweenness}
    \begin{tabular}{lrrrrrrrr}
	\toprule
	& \multicolumn{4}{c}{Copenhagen calls} & \multicolumn{4}{c}{Copenhagen SMS} \\ \cmidrule(l){2-5}\cmidrule(l){6-9}
	weighted                        & no      & yes     & no      & yes     & no       & yes      & no       & yes      \\
	directed                        & yes     & yes     & no      & no      & yes      & yes      & no       & no       \\
	$|V|$                           & 536     & 536     & 536     & 536     & 568      & 568      & 568      & 568      \\
	$|E|$                           & 924     & 924     & 621     & 621     & 1303     & 1303     & 697      & 697      \\
	$\Delta(G)$                     & 18      & 18      & 18      & 18      & 11       & 11       & 11       & 11       \\
	$|\mathcal{SP}(G)|/U(G)$        & 1.43    & 1.35    & 1.79    & 1.50    & 1.99     & 1.30     & 2.12     & 1.32     \\
	diam                            & 21      & 75      & 22      & 197     & 22       & 1783     & 20       & 3781     \\  
        \midrule
        \multicolumn{9}{c}{Betweenness centrality} \\    
        \midrule
	diam nodes traversed            & 22      & 25      & 23      & 15      & 23       & 10       & 21       & 12       \\
	diam centrality                 & 54084   & 35826   & 133280  & 140648  & 127129   & 81092    & 224086   & 118298   \\
	path length                     & 9       & 17      & 7       & 32      & 8        & 81       & 8        & 106      \\
	path nodes traversed            & 10      & 11      & 8       & 16      & 9        & 16       & 9        & 23       \\
	path centrality                 & 59959   & 57692   & 177592  & 154253  & 228071   & 156347   & 280440   & 159783   \\	
	preprocessing time              & 0.06    & 0.11    & 0.10    & 0.26    & 0.14     & 0.45     & 0.15     & 0.55     \\
    Alg.~\ref{alg:betweenness-algorithm} runtime                  & 1118.62 & 1182.93 & 5320.75 & 5741.69 & 12492.22 & 13737.50 & 15609.16 & 17860.28 \\	
    \bottomrule
    \end{tabular}
\end{table*}

\begin{table*}[htb]
    \centering
    \footnotesize
    \caption{Results for trust graph instances. Note that $|E| = w_{\text{sum}}$ for the augmented graphs. Runtimes are in seconds. Rows \textit{path length} and \textit{path centrality} give the length and the centrality of the optimal shortest path, respectively.}
    \label{tab:weighted-trust}
    \begin{tabular}{lrrrrrr}
	\toprule
	& \multicolumn{2}{c}{Bitcoin Alpha} & \multicolumn{2}{c}{Bitcoin OTC} & \multicolumn{2}{c}{Advogato} \\ \cmidrule(l){2-3}\cmidrule(l){4-5}\cmidrule(l){6-7}
	weighted                 & no      & yes      & no      & yes      & no      & yes      \\
	$|V|$                    & 3783    & 281050   & 5881    & 397821   & 6539    & 167369   \\
	$|E|$                    & 24186   & 301453   & 35592   & 427532   & 51127   & 211957   \\
	$\Delta(G)$              & 511     & 511      & 795     & 795      & 805     & 805      \\
	$|\mathcal{SP}(G)|/U(G)$ & 9.50    & 2.68     & 10.74   & 19.18    & 8.54    & 1.83     \\
	diam                     & 10      & 107      & 11      & 107      & 11      & 35       \\  
        \midrule
        \multicolumn{7}{c}{Degree centrality} \\    
        \midrule
	diam nodes traversed     & 11      & 12       & 12      & 12       & 12      & 7        \\
	diam centrality          & 689     & 266      & 971     & 299      & 318     & 846      \\
	path length              & 5       & 39       & 4       & 31       & 4       & 11       \\
	path nodes traversed     & 6       & 9        & 5       & 12       & 5       & 4        \\
	path centrality          & 865     & 891      & 1349    & 1433     & 1318    & 1165     \\
	Alg.~\ref{alg:algorithm} runtime           & 1134.96 & 20996.75 & 4189.95 & 58050.34 & 1345.74 & 24927.09 \\ \bottomrule
    \end{tabular}    
\end{table*}

\subsection{Discussion and insights}

We observe that using \Cref{alg:algorithm}, we can solve Problem~\eqref{eq:degree-problem} much more efficiently than by using the MVP algorithm. For instance, the Watts-Strogatz instances with 1000 nodes, which took around 8 minutes to solve, can now be solved in roughly 10 seconds. Problems involving 5000 and 10000 nodes can now be solved within a fifth and half of an hour, respectively. A similar observation can be made for the Barab\'{a}si-Albert graphs. We did not consider solving larger instances with the MVP algorithm, as one 5000-node instance took over 13 hours to solve. This is consistent with the expected runtime based on the difference between solving 100- and 500-node instances. For real-world instances, the improvement in runtime is also significant. For example, the Cerevisiae instance no longer requires over twenty minutes to solve; it can be solved in just over twenty seconds. We also note that in practice, \Cref{alg:algorithm} scales much more efficiently than stated in \Cref{thm:runtime} because, for most vertices $u$ and paths $P_u$, the size of $\mathcal{P}_u$ and $\mathcal{N}(P_u)$ is small as the number of immediate neighbours is smaller than the maximum degree of the graph $\Delta(G)$. We also observe the difference in running time over the two types of synthetic instances. Although the number of vertices and edges in both sets of graphs are nearly identical, the differences in running times for both algorithms on both sets of graphs are pronounced. \Cref{alg:algorithm} returns a solution faster in Watts-Strogatz graphs than Barab\'{a}si-Albert graphs of similar sizes, while the converse is true for the MVP algorithm. This is consistent with the theoretical running times of the MVP algorithm, and \Cref{thm:runtime} as Watts-Strogatz graphs have smaller maximum degrees but larger diameters than Barab\'{a}si-Albert graphs. 

We also observe that Problem~\eqref{eq:degree-problem} on a weighted graph is much harder to solve. However, as mentioned earlier, we see that the ratio of the running times between the weighted and unweighted case is not $w_{\text{sum}} / |E|$ as implied by the analysis. It is larger. This is because the work done in each iteration to compute and compare $C_d(P)$ is not the same across both cases. More work needs to be done in the weighted case as we need to check the type of node to which we are extending the path. \Cref{tab:weighted-copenhagen} also offers insight into \Cref{alg:algorithm}'s performance on directed and undirected graphs. In order to model the undirected edges, we use two directed edges. This means that although undirected graphs have fewer edges, the algorithm is faster on directed graphs as there is not necessarily bidirectional flow between all pairs of connected vertices. In addition, we can see that the degree centrality of the optimal path is different when weights are considered. There is no guarantee on whether this value is higher or lower. Hence, there is a need to compare the cases separately depending on the problem instance.

On the other hand, finding the most betweenness-central shortest path requires significantly higher computational effort. \Cref{alg:betweenness-algorithm} requires less than 1 minute, on average, to solve Problem~\eqref{eq:betweenness-path-problem} on both Watts-Strogatz and Barab\'{a}si-Albert graphs with 100 nodes, but when the number of vertices is increased to 500, the running time exceeds six hours. This result is mirrored in \Cref{tab:real} where we again see that finding the most betweenness-central shortest path in our real-world instances is considerably more time-consuming than finding the most degree-central shortest path. For example, the Bus instance requires over 21 hours to solve, compared to just under 4 seconds if degree centrality were to be used instead. This result arises because path betweenness centrality is a global measure of centrality that takes into account the entire topology of the graph, whereas path degree centrality is a local measure of centrality that only considers the immediate neighbours of the path. Naturally, this means that computing the betweenness centrality will take much longer than the degree centrality.

Similar to the degree centrality case, \Cref{tab:weighted-copenhagen-betweenness} allows us to compare the performance of our algorithm on weighted graphs. Again, we see that the weighted case takes more time to solve, but the difference is less pronounced than when the degree centrality was considered. This is because the input graph does not need to be manipulated, and the main difference in \Cref{alg:betweenness-algorithm} when solving Problem~\eqref{eq:betweenness-path-problem} on a weighted graph is an additional work required to insert into a priority queue. Furthermore, \Cref{tab:weighted-copenhagen-betweenness} also allows us to see the difference in the performance of the algorithm when comparing directed and undirected graphs. Again, we notice that the running time is much greater for the undirected instances, as we need to consider two directed edges for each undirected edge.

Note that in line 14 of \Cref{alg:betweenness-algorithm}, we compare the betweenness centrality of two paths, $\langle P_u, v\rangle$ and $P_v$. This can be accomplished in two ways: (1) by computing the betweenness centrality of the associated paths ``on the fly'' or (2) by performing preprocessing as suggested in \citet{Brandes2008} and computing the centralities by adding accumulated dependencies as needed; see \Cref{remark} for more details. Naturally, we opt for option (2) because it is more efficient. Therefore, to provide the full picture of the algorithm's performance, we report the preprocessing time for each network instance; see row \textit{preprocessing time} in Tables~\ref{tab:ws1}-\ref{tab:real} and Table~\ref{tab:weighted-copenhagen-betweenness}. Our results indicate that the time required for preprocessing is negligible relative to the running time of \Cref{alg:betweenness-algorithm}, with preprocessing taking less than a second on the largest graphs considered.

We finally note that although in our experiments all three algorithms were run sequentially over all starting vertices, they could be run in parallel for an improved runtime. For example, solving the weighted instance of Bitcoin OTC with Algorithm 1 on 58 cores will take approximately 1000 seconds; see Table~\ref{tab:weighted-trust}.

\section{Concluding remarks}\label{sec:conclusion}

In this paper, we consider the problem of finding the most central shortest path in a graph. We show that the computational complexity of this problem depends on the measure of centrality employed. The problem of finding the most betweenness-central shortest path is solvable in polynomial time for both unweighted and weighted graphs. On the other hand, the problem of finding the most degree-central shortest path can be solved in polynomial time on unweighted graphs, but is NP-hard on weighted graphs. In contrast, the problem of finding the most closeness-central shortest path is NP-hard on both unweighted and weighted graphs. We hypothesise that this relationship arises as betweenness centrality is distance-invariant, i.e., we are just counting the number of shortest paths. Similarly, finding the most degree-central shortest path in an unweighted graph can be computed in polynomial time because optimal paths of length $k$ can be obtained by extending optimal paths of length $k-2$. However, when weights are added to the edges, the result no longer holds, and the problem becomes NP-hard. Unlike these two metrics, closeness centrality has a built-in distance measure, and hence, the problem of finding the most closeness-central shortest path is NP-hard regardless of whether the graph has weighted edges or not.

We summarise the complexity status of the problems considered in this paper in \Cref{tab:complexity-summary}. Although \Cref{tab:complexity-summary} suggests that finding the most betweenness-central shortest path is less complex than finding the most degree-central shortest path, our computational results indicate otherwise, even when considering graphs with positively weighted edges. We believe the discrepancy is caused by the fact that path degree centrality is a local measure of centrality, while path betweenness centrality is a global measure of centrality that takes significantly more time to compute. In practice, this means that a planner should decide whether it is worth the computational effort to find the most betweenness-central shortest path or if the most degree-central shortest path could be used as an adequate proxy.

\begin{table}[!hbtb]
    \centering
    \caption{Comparison of the complexity status of the problems considered on unweighted and positive edge-weighted graphs.}
    \label{tab:complexity-summary}
    \begin{tabular}{lcc}
    \toprule
    Centrality measure & Unweighted graph & Weighted graph \\ \midrule
    Betweenness        & P                & P              \\
    Degree             & P                & NP-hard        \\
    Closeness          & NP-hard          & NP-hard        \\ \bottomrule
    \end{tabular}
\end{table}

Future work in this area may focus on studying the problem of finding the most closeness-central shortest path with other aggregation functions or investigating MIP formulations, heuristics, and approximation schemes for problems established to be NP-hard in this paper.

\bibliographystyle{apalike} 

\bibliography{refs}

\end{document}